\documentclass{lmcs} %%% last changed 2014-08-20
\pdfoutput=1
\usepackage[utf8]{inputenc}

\usepackage{lastpage}
\lmcsdoi{19}{4}{36}
\lmcsheading{}{\pageref{LastPage}}{}{}%
{Jun.~30,~2022}{Dec.~21,~2023}{}


\usepackage{hyperref}
\usepackage{amsmath,amsfonts,latexsym}
\usepackage[absolute,overlay]{textpos} 
\usepackage{graphicx}
\usepackage{float}
\usepackage{enumerate}
\usepackage{multirow}
\usepackage{multicol}
\usepackage{mathtools}
\usepackage{amssymb}
\usepackage{braket}
\usepackage{psfrag, epstopdf}
\usepackage{cite}
\usepackage{algorithm2e}
\usepackage[noend]{algorithmic}
\usepackage{tabu}
\usepackage{tikz}
\usepackage{caption}
\usetikzlibrary{automata, positioning, arrows}

\newcommand{\N}{{\mathbb{N}}}  % the natural numbers
\newcommand{\Z}{{\mathbb{Z}}}  % the integers
\newcommand{\Q}{{\mathbb{Q}}} % the rational numbers
\newcommand{\R}{{\mathbb{R}}}  % the real numbers
\newcommand{\pfa}{{\mathcal{P}}}  % tthe empty word

\newcommand{\0}{\mathbf{0}} % Zero matrix
\begin{document}

\title[Decision Questions for Probabilistic Automata on Small Alphabets]{Decision Questions for Probabilistic\texorpdfstring{\\}{ }Automata on Small Alphabets}
\titlecomment{{\lsuper*}This is an extended version of the conference paper \cite{BS21}.}
%\thanks{thanks, optional.}	%optional

% affiliations are numbered automatically with a, b, c (see below)
% use the optional argument to indicate the affiliation(s) of each author
% omit the argument if there is only one author, or only one affiliation
\author[P.~C.~Bell]{Paul C.~Bell\lmcsorcid{0000-0003-2620-635X}}[a]
\author[P. Semukhin]{Pavel Semukhin\lmcsorcid{0000-0002-7547-6391}}[b]

% affiliation 1 (automatically numbered a)
\address{Keele University, School of Computer Science and Mathematics, Colin Reeves Building, Keele, Staffordshire, ST5 5BG, UK}	%optional
% write emails for all authors having that affiliation
\email{p.c.bell@keele.ac.uk}  %optional

% affiliation 2 (automatically numbered b)
\address{Liverpool John Moores University, School of Computer Science and Mathematics, James Parsons Building, Byrom Street, Liverpool, L3 3AF, UK}	%optional
\email{P.Semukhin@ljmu.ac.uk}  %optional

%% etc.

%% required for running head on odd and even pages, use suitable
%% abbreviations in case of long titles and many authors:

%%%%%%%%%%%%%%%%%%%%%%%%%%%%%%%%%%%%%%%%%%%%%%%%%%%%%%%%%%%%%%%%%%%%%%%%%%%

%% the abstract has to PRECEDE the command \maketitle:
%% be sure not to issue the \maketitle command twice!

\begin{abstract}
  \noindent We study the emptiness and $\lambda$-reachability problems for unary and binary Probabilistic Finite Automata (PFA) and characterise the complexity of these problems in terms of the degree of ambiguity of the automaton and the size of its alphabet. Our main result is that emptiness and $\lambda$-reachability are solvable in EXPTIME for polynomially ambiguous unary PFA and if, in addition,  the transition matrix is over $\{0, 1\}$, we show they are in NP. In contrast to the Skolem-hardness of the $\lambda$-reachability and emptiness problems for exponentially ambiguous unary PFA, we  show that these problems are NP-hard even for finitely ambiguous unary PFA. We also show that the value of a polynomially ambiguous PFA can be computed in EXPTIME. For binary polynomially ambiguous PFA with commuting transition matrices, we prove NP-hardness of the $\lambda$-reachability (dimension 9), nonstrict emptiness (dimension 37) and strict emptiness (dimension 40) problems. 
\end{abstract}

\maketitle

\section{Introduction}
% General Intro
There are many possible extensions of the fundamental notion of a nondeterministic finite automaton. One such notion is that of  Probabilistic Finite Automata (PFA) which was first introduced by Rabin  \cite{Ra63}. In a PFA $\mathcal{P}$ over a (finite) input alphabet $\Sigma$ the outgoing transitions from a state, for each input letter of $\Sigma$, form a probability distribution, as does the initial state vector. Thus, a word $w \in \Sigma^*$ is accepted with a certain probability, which we denote $\mathcal{P}(w)$.

% Definition of problems
There are a variety of interesting questions that one may ask about a PFA $\mathcal{P}$ over an alphabet $\Sigma$. In this article we focus on two decision questions, that of \emph{$\lambda$-reachability} and \emph{emptiness}. The $\lambda$-reachability problem is stated thus: given a probability $\lambda \in [0, 1]$, does there exist some word $w \in \Sigma^*$ such that $\mathcal{P}(w) = \lambda$? In the (strict) emptiness problem, we ask if there exists \emph{any} word $w \in \Sigma^*$ such that $P(w) \geq \lambda$ (resp. $P(w) > \lambda$). We also mention the related \emph{cutpoint isolation} problem --- to determine if for each $\epsilon > 0$,  there exists a word $w \in \Sigma$ such that $|\mathcal{P}(w) - \lambda| < \epsilon$. The \emph{value}-$1$ problem is a special case of the cutpoint isolation problem when $\lambda = 1$ \cite{GO10} and the \emph{value} problem is to determine the supremum over all words of the acceptance probability of a word.

% Known results pre ambiguity
In general, the emptiness problem is undecidable for PFA \cite{Paz}, even over a binary alphabet when the automaton has $25$ states \cite{Hir}. The cutpoint isolation problem is undecidable \cite{BMT77} even for PFA with $420$ states over a binary alphabet \cite{BC03}. %The cutpoint isolation problem, in the special case where $\lambda = 1$ (the value-$1$ problem), is also undecidable \cite{GO10}. 
The problem is especially interesting given the seminal result of Rabin that if a cutpoint $\lambda$ is isolated, then the cutpoint language associated with $\lambda$ is necessarily regular \cite{Ra63}. 

% Definition of ambiguity
We may ask which restrictions of PFA may lead to decidability of the previous problems. In this paper we are interested in PFA of \emph{bounded ambiguity}, where the ambiguity of a word denotes the number of accepting runs of that word in the PFA. A PFA  $\mathcal{P}$ is $f$-ambiguous,
for a function $f : \mathbb{N} \to \mathbb{N}$, if every word of length $n$ has at most $f(n)$ accepting
runs. A run is  \emph{accepting} if the probability of that run ending in a final state is strictly positive. The degree of ambiguity is thus a property of the NFA underlying a PFA (i.e.,~the NFA produced by  setting all nonzero transition probabilities to $1$). We may consider the notions of finite, polynomial or exponential ambiguity of $\mathcal{P}$ based on whether $f$ is bounded by a constant, is a polynomial or else is exponential, respectively. Characterisations of the degree of ambiguity of NFA are given by Weber and Seidel \cite{WS91} and we discuss more about this in Section~\ref{amb-sec}. 

% Background on results of PFA ambiguity
The authors of \cite{FR17} show that emptiness for PFA remains undecidable even for polynomially ambiguous automata (quadratic ambiguity), show PSPACE-hardness results for finitely ambiguous PFA and that emptiness is in NP for the class of $k$-ambiguous PFA for every $k > 0$. The emptiness problem for PFA was later shown to be undecidable for linearly ambiguous automata \cite{DJ18}. The emptiness problem is known to be in quasi-polynomial time for $2$-ambiguous PFA \cite{FR17}. Furthermore, the value of a $k$-ambiguous PFA is approximable in polynomial time up to any multiplicative constant \cite{FR17}.

Another restriction is to constrain input words of the PFA to come from a given language~$\mathcal{L}$. If $\mathcal{L}$ is a \emph{letter-bounded} language\footnote{A language \(\mathcal{L}\) over an alphabet \(\Sigma = \{a_1, \ldots, a_k\}\) is \emph{letter-bounded} if \(\mathcal{L} \subseteq a_{j_1}^* \cdots a_{j_p}^*\), where \(1 \leq j_i \leq k\) for \(1 \leq i \leq p\).}, then the emptiness and $\lambda$-reachability problems remain undecidable for polynomially ambiguous PFA, even when all transition matrices commute \cite{Bell19}. In contrast, the \emph{cutpoint-isolation} problem is decidable even for exponentially ambiguous PFA when inputs are constrained to come from a given letter-bounded context-free language, although it is NP-hard for $3$-state PFA on letter-bounded inputs \cite{BS20}.

Our main results are as follows. We show that the $\lambda$-reachability and emptiness problems for probabilistic finite automata are:
\begin{itemize}
    \item In EXPTIME for the class of polynomially ambiguous unary PFA and are NP-complete if, in addition, the transition matrix is over $\{0,1\}$ [Theorem \ref{mainthm} and Corollary \ref{cor:NPcmpl}].
    \item NP-hard for polynomially ambiguous PFA over a binary alphabet with \emph{fixed} and \emph{commuting} transition matrices of dimension $40$ (strict emptiness problem), $37$ (nonstrict emptiness problem) and $9$ ($\lambda$-reachability problem) [Theorem \ref{thm:NPbin}].
\end{itemize}

We also show NP-hardness for the class of finitely ambiguous unary PFA with $\{0,1\}$ transition matrix [Proposition~\ref{thm:NPhard}]. Our hardness results rely on the NP-hardness of solving binary quadratic equations and the universality problem for unary regular expressions. Note that these restrictions of the PFA, to have polynomial ambiguity, a binary alphabet, fixed and commuting transition matrices, or else finite ambiguity, unary alphabet and binary transition matrix, makes it more difficult to prove an NP-hard lower-bound. The aim is to have as restricted a model as possible, while retaining NP-hardness.     

An interesting question, that is left open, is to find the exact  complexity of these problems in the case of polynomially ambiguous unary PFA, i.e.\ to close the gap between the EXPTIME upper bound and NP-hard lower bound.

There are connections between problems on polynomially ambiguous unary PFA and reachability problems for special types of linear recurrence sequences (LRS).
For example, polynomially ambiguous unary PFA (or more generally, polynomially ambiguous unary weighted automata) describe LRS whose characteristic roots are roots of rational numbers \cite{BFL22}. It was recently shown that for LRS whose characteristic roots are roots of real algebraic numbers, the reachability problem (also called the Skolem problem) can be solved in \(\textrm{NP}^{\textrm{RP}}\) \cite{ABM20}. However, the translation given in \cite{BFL22} can produce a description for an LRS that is exponential in the size of the PFA. On the other hand, it is claimed in \cite{ABM20} that the reachability problem for PFA is in \(\textrm{NP}^{\textrm{RP}}\) for the special case where roots are distinct roots of real numbers. Our results do not have such restrictions.

Our proof of the EXPTIME upper bound relies on a careful analysis of how fast the acceptance probability $\mathcal{P}(a^k)$ converges to its limit values as $k\to \infty$. To prove the NP bound when the transition matrix is over $\{0,1\}$, we consider a nondeterministic version of the EXPTIME algorithm, showing that the verification part can be done in polynomial time. We also show that the value of a polynomially ambiguous PFA can be computed in EXPTIME as a corollary of our techniques.

This is an extended version of the conference paper \cite{BS21}. The present version contains an expanded introduction, additional figures and explanations and full proofs which were in the appendix of the submitted conference version. We also consider the problem of computing the value of a PFA which was not considered in \cite{BS21}. 

\section{Probabilistic Finite Automata and Notation}

We denote by $\mathbb{Q}^{n \times n}$ the set of all $n \times n$ matrices over  $\mathbb{Q}$. Given two column vectors $u \in \mathbb{Q}^n$ and $v \in \mathbb{Q}^m$, we denote by $[u|v]$ the column vector $(u_1, \ldots, u_n, v_1, \ldots, v_m)^T \in \mathbb{Q}^{n+m}$. For a sequence of vectors $u_1, u_2, \ldots, u_k$, we write $[u_1 | u_2 | \cdots | u_k]$ for the column vector which stacks the vectors on top of each other. 

Given $A = (a_{ij}) \in\Q^{m\times m}$ and $B\in\Q^{n\times n},$ we define the direct sum $A\oplus B$ and Kronecker product $A \otimes B$ of $A$ and $B$ by:
\[
A\oplus B=
\left[\begin{array}{@{}c|l@{}}
A & \0_{m,n}\\
\hline
\0_{n,m} & B
\end{array}\right], \quad
A\otimes B=
\left[\begin{array}{cccc}
a_{11}B & a_{12}B & \cdots & a_{1m}B \\ 
a_{21}B & a_{22}B & \cdots & a_{2m}B \\ 
\vdots & \vdots & \ddots &\vdots \\ 
a_{m1}B & a_{m2}B & \cdots & a_{mm}B \\ 
\end{array}\right],
\]
where $\0_{i,j}$ denotes the zero matrix of dimension $i \times j$. Note that neither $\oplus$ nor $\otimes$ are commutative in general. The following useful properties of $\oplus$ and $\otimes$ are well known. 

\begin{lem}\label{kronprop}
Let $A, B, C, D \in \mathbb{Q}^{n \times n}$. Then we have:
\begin{itemize}
\item Associativity: $(A \otimes B) \otimes C = A \otimes (B \otimes C)$ and $(A \oplus B) \oplus C = A \oplus (B \oplus C)$, thus $A \otimes B \otimes C$ and $A \oplus B \oplus C$ are unambiguous.
\item Mixed product properties: $(A \otimes B)(C \otimes D) = (AC \otimes BD)$ and $(A \oplus B)(C \oplus D) = (AC \oplus BD)$.
\item If $A$ and $B$ are stochastic matrices, then so are $A \oplus B$ and $A \otimes B$. 
\item If $A, B \in \mathbb{Q}^{n \times n}$ are both upper-triangular, then  so are $A \oplus B$ and $A \otimes B$.
\end{itemize}
\end{lem}

See \cite{HJ91} for proofs of the first three properties of Lemma~\ref{kronprop}. The fourth property follows directly from the definition of the direct sum and Kronecker product.

A Probabilistic Finite Automaton (PFA) $\mathcal{P}$ with $n$ states over an alphabet $\Sigma$ is defined as $\mathcal{P}=(u, \{M_a|a \in \Sigma\}, v)$ where $u \in \mathbb{Q}^n$ is the \emph{initial probability distribution}; $v \in \{0, 1\}^n$ is the \emph{final state vector} and each $M_a \in \mathbb{Q}^{n \times n}$ is a (row) stochastic matrix. We will primarily be interested in \emph{unary} and \emph{binary} PFA, for which $|\Sigma| = 1$ and $|\Sigma| = 2$ respectively. For a word $w = a_1a_2\cdots a_k \in \Sigma^*$, we define the acceptance probability $\pfa(w): \Sigma^* \to \mathbb{Q}$ of $\pfa$ as:
\[
\pfa(w) = u^T M_{a_1}M_{a_{2}} \cdots M_{a_k} v \in [0, 1],
\]
which denotes the acceptance probability of $w$.\footnote{Some authors interchange the order of ${u}$ and ${v}$ and use column stochastic matrices, although the two definitions are trivially equivalent.} 

For a given cutpoint $\lambda \in [0, 1]$, we define the following languages: $L_{\geq \lambda}(\pfa) = \{w \in \Sigma^*\ |\ \pfa(w) \geq \lambda\}$, a nonstrict cutpoint language, and $L_{> \lambda}(\pfa) = \{w \in \Sigma^*\ |\ \pfa(w) > \lambda\}$, a strict cutpoint language. The (strict) emptiness problem for a cutpoint language is to determine if $L_{\geq \lambda}(\pfa) = \emptyset$ (resp. $L_{> \lambda}(\pfa) = \emptyset$). We are also interested in the \emph{$\lambda$-reachability} problem, for which we ask if there exists a word $w \in \Sigma^*$ such that $\mathcal{P}(w) = \lambda$. 

\subsection{PFA Ambiguity}\label{amb-sec}

The degree of ambiguity of an automaton is a structural parameter, indicating the number of accepting runs for a given input word. See \cite{WS91} for details of ambiguity for nondeterministic automata and \cite{Bell19, BS20, DJ18, FR17} for connections to PFA. 

Let $w \in \Sigma^*$ be an input word of an NFA $\mathcal{N} = (Q, \Sigma, \delta, Q_I, Q_F)$, with $Q$ the set of states, $\Sigma$ the input alphabet, $\delta \subset Q \times \Sigma \times Q$ the transition function, $Q_I$ the set of initial states and $Q_F$ the set of final states. For each $(p, w, q) \in Q \times \Sigma^* \times Q$, define $\textrm{da}_{\mathcal{N}}(p, w, q)$ as the number of paths for $w$ in $\mathcal{N}$ leading from state $p$ to $q$. The \emph{degree of ambiguity} of $w$ in $\mathcal{N}$, denoted $\textrm{da}_{\mathcal{N}}(w)$,  is defined as the number of all \emph{accepting paths} for $w$ (starting from an initial and ending in a final state). The \emph{degree of ambiguity} of $\mathcal{N}$, denoted $\textrm{da}(\mathcal{N})$, is the supremum of the set $\{\textrm{da}_{\mathcal{N}}(w)\ |\ w \in \Sigma^*\}$. $\mathcal{N}$ is called infinitely ambiguous if $\textrm{da}(\mathcal{N}) = \infty$, finitely ambiguous if $\textrm{da}(\mathcal{N}) < \infty$, and unambiguous if $\textrm{da}(\mathcal{N}) \leq 1$. The \emph{degree of growth} of the ambiguity of $\mathcal{N}$, denoted $\textrm{deg}(\mathcal{N})$, is defined as the minimum degree of a univariate polynomial $h$ with positive integral coefficients such that for all $w \in \Sigma^*$, $\textrm{da}_{\mathcal{N}}(w) \leq h(|w|)$ (if such a polynomial exists, in which case $\mathcal{N}$ is called polynomially ambiguous, otherwise the degree of growth is infinite and $\mathcal{N}$ which we call exponentially ambiguous). 

The above notions relate to NFA. We may derive an analogous notion of ambiguity for PFA by considering an embedding of a PFA $\pfa$ into an NFA $\mathcal{N}$ in such a way that for each letter $a \in \Sigma$, if the probability of transitioning from a state $i$ to state $j$ is nonzero under $\pfa$, then there is an edge from state $i$ to $j$ under $\mathcal{N}$ for letter $a$. The initial states of $\mathcal{N}$ are those of $\pfa$ having nonzero initial probability and the final states of $\mathcal{N}$ and $\pfa$ coincide. We then say that $\pfa$ is finitely/polynomially/exponentially ambiguous if $\mathcal{N}$ is (respectively). 

A state $q \in Q$ in an NFA (resp. PFA) is called \emph{useful} if there exists an accepting path which visits $q$ (resp. an accepting path of nonzero probability which visits $q$). We can characterise whether an NFA $\mathcal{A}$ (and thus a PFA by the above embedding) has finite, polynomial or exponential ambiguity using the following properties (see Figs~1 and 2): 

\noindent {\bf EDA} - There is a useful state $q \in Q$ such that, for some word $v \in \Sigma^*$, $da_{\mathcal{A}}(q, v, q) \geq 2$.

\noindent {\bf $\text{IDA}_d$} - There are useful states $r_1, s_1, \ldots, r_d, s_d \in Q$ and words $v_1, u_2, v_2, \ldots, u_d, v_d \in \Sigma^*$ such that for all $1 \leq i \leq d$, $r_i$ and $s_i$ are distinct and $(r_i, v_i, r_i), (r_i, v_i, s_i), (s_i, v_i, s_i) \in \delta$ and for all $2 \leq i \leq d$, $(s_{i-1}, u_{i}, r_{i}) \in \delta$. 

\noindent\begin{minipage}{\textwidth}
\begin{minipage}[c][6cm][c]{0.38\textwidth}
  \begin{tikzpicture}[shorten >=1pt,node distance=1.45cm,on grid,auto] 
   \node[state] (q)   {$q$}; 
   \node[state, initial, initial text = {}] (q_0) [left=of q] {$q_0$}; 
   \node[state, accepting](q_F) [right=of q] {$q_F$};
    \path[->] 
    (q_0) 		edge [above, dashed] node {$w_1$} (q)
    (q)  edge [loop below] node {$v$} ()
    (q)  edge [loop above] node {$v$} ()
edge [above, dashed] node {$w_2$} (q_F);
           %(q_2) edge [loop below] node {$\{0,1\}:1$} ();
    %(q_2) edge  node [swap] {0} (q_3) 
     %     edge [loop below] node {1} ();
\end{tikzpicture}
\label{fig_eda}

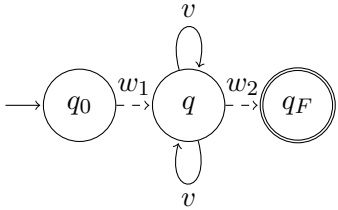
\captionof{figure}{EDA Property}
%\caption{EDA Property}
\end{minipage}\hfill
\begin{minipage}[c][6cm][c]{0.62\textwidth}
\begin{tikzpicture}[shorten >=1pt,node distance=1.4cm,on grid,auto] 
   \node[state] (p_1)   {$r_1$};
   \node[state] (p_2) [right=of p_1]  {$s_1$};
   \node[] (p_2b) [right=of p_2]  {};
   \node[] (p_2c) [right=of p_2b]  {};
   \node[state] (p_3) [right=of p_2c]  {$r_d$};
   \node[state] (p_4) [right=of p_3]  {$s_d$};
   \node[state, initial, initial text = {}] (q_0) [below=of p_1] {$q_0$}; 
   \node[state, accepting](q_F) [below=of p_4] {$q_F$};
    \path[->] 
    (p_1)  edge [loop above] node {$v_1$} ()
    (p_2)  edge [loop above] node {$v_1$} ()
    (p_3)  edge [loop above] node {$v_d$} ()
    (p_4)  edge [loop above] node {$v_d$} ()
    (p_4)  edge [left] node {$w_2$} (q_F)
    (p_1)  edge [above] node {$v_1$} (p_2)
    (p_3)  edge [above] node {$v_d$} (p_4)
    (p_2c)  edge [above] node {$u_d$} (p_3)
    (p_2)  edge [above] node {$u_2$} (p_2b)
    (p_2b)  edge [above, dashed] node {} (p_2c)
    (q_0)  edge [right] node {$w_1$} (p_1);
%edge [below] node {$a^y$} (q_F);
           %(q_2) edge [loop below] node {$\{0,1\}:1$} ();
    %(q_2) edge  node [swap] {0} (q_3) 
     %     edge [loop below] node {1} ();
\end{tikzpicture}
\label{fig_idad}
\captionof{figure}{IDA$_d$}
\end{minipage}%
\end{minipage}

\begin{thmC}[\cite{IR86, Re77, WS91}]\label{crtitthm}
An NFA (or PFA) $\mathcal{A}$ having the EDA property is equivalent to it being exponentially ambiguous. For any $d \in \mathbb{N}$, an NFA (or PFA) $\mathcal{A}$ having property $\text{IDA}_d$ is equivalent to $\textrm{deg}(\mathcal{A}) \geq d$.
\end{thmC}

Clearly, if $\mathcal{N}$ agrees with $\text{IDA}_d$ for some $d > 0$, then it also agrees with $\text{IDA}_1, \ldots, \text{IDA}_{d-1}$. An NFA (or PFA) is thus finitely ambiguous if it does not possess property $\text{IDA}_1$.

\section{Unary PFA}

Our main focus is on unary automata. We begin by giving a simple folklore proof that the $\lambda$-reachability and emptiness problems are as computationally difficult as the famous Skolem problem, which is only known to be decidable for instances of depth $4$ \cite{Vereshchagin1985}. See also \cite{AAOW15} for connections to reachability problems for Markov chains. 

\begin{thm}\label{skolRes}
The $\lambda$-reachability and emptiness problems for unary exponentially ambiguous Probabilistic Finite Automata are Skolem-hard.
\end{thm}

\let\oldproofname\proofname\renewcommand{\proofname}{\textit{Proof (Folklore).}}
\begin{proof} %\emph{(Folklore).}
The $\lambda$-reachability problem for unary exponentially ambiguous PFA can be shown Skolem-hard  based on the matrix formulation of Skolem's problem \cite{HHHK05} and Turakainen's technique showing the equivalence of (strict) cutpoint language acceptance of generalised automata and exponentially ambiguous probabilistic automata \cite{Tu69}. We shall now show the details of this procedure. 

Consider the following linear recurrence sequence (LRS) $u_n$ for $n \geq 0$, with recurrence coefficients $a_0, \ldots, a_{k-1} \in \Z$ and where $u_0, \ldots, u_{k-1} \in \Z$ are the given initial conditions:
\[
u_n = a_{k-1}u_{n-1} + \cdots + a_0u_{n-k}
\]
We denote the zero set of $u_n$ by $Z(u_n) = \{j | u_j = 0\}$. Skolem's problem is to determine if $Z(u_n)$ is empty for a given LRS. The value $k$ is the \emph{depth} of the LRS. It is decidable to determine if $Z(u_n)$ is infinite \cite{Hansel}.

Let $M \in \Z^{k \times k}$ be defined thus:
\[
M = \left( \begin{array}{ccccc} a_{k-1} & 1 & \cdots & 0 & 0 \\ \vdots & \vdots & \ddots & \vdots & \vdots \\ a_2 & 0 & \cdots & 1 & 0 \\ a_1 & 0 & \cdots & 0 & 1 \\ a_0 & 0 & \cdots & 0 & 0 \end{array}\right),
\]
and let $u' = (u_{k-1}, \ldots, u_0)^T \in \Z^{k}$ and $v' = (0, \ldots, 0, 1) \in \Z^k$. It is not difficult to verify that $u_j = u'^T M^j v'$.

We may now define:
\[
A' = \left( \begin{array}{ccc} 0 & {\bf 0} & 0 \\ {\bf s}^T & M & {\bf 0}^T \\ r & {\bf t} & 0 \end{array}\right) \in \Z^{(k+2) \times (k+2)},
\]
where $r \in \mathbb{Z}$, ${\bf s}, {\bf t} \in \Z^{k}$ are (uniquely) chosen so that each row and column sum is zero and ${\bf 0} \in \Z^k$ is the zero vector. Clearly then $A'^j$ also has zero row and column sums and retains this same structure for any $j > 0$. Define $\Omega \in \Z^{(k+2) \times (k+2)}$ such that $\Omega_{i, \ell} = 1$ for $1 \leq i, \ell \leq k+2$. Notice that $\Omega^j = (k+2)^{j-1} \Omega$ and $A'\Omega = \Omega A' = {\bf 0}$. Let $c = \max\{|a_\ell| | 1 \leq \ell \leq k-1\}+1$ and note that $A'+c\Omega$ is strictly positive. Define $A = \frac{1}{c(k+2)}(A'+c\Omega)$, $u = \frac{[0 | u' | 0]}{|u'|_1} \in \Q^{k+2}$ and $v = [0|v'|0] \in \Z^{k+2}$, noting that $|u|_1$ = 1. Let our PFA be given by $\pfa = (u, A, v)$.

It can now be seen that:
\begin{eqnarray*}
\pfa(a^j) & = & u^TA^j v \\
 & = & \left(\frac{1}{c(k+2)}\right)^j u^T(A'+c\Omega)^jv \\
 & = & \left(\frac{1}{c(k+2)}\right)^j u^T(A'^j+c^j\Omega^j)v \\
 & = & \left(\frac{1}{c(k+2)}\right)^j(u'^TM^jv' + c^j(k+2)^{j-1}u^T\Omega v) \\
 & = & \left(\frac{1}{c(k+2)}\right)^j u_j + \frac{1}{(k+2)},
\end{eqnarray*}
which equals $\frac{1}{(k+2)}$ if and only if $u_j = 0$. Thus determining if cutpoint $\lambda = \frac{1}{(k+2)}$ can be reached is Skolem-hard.

The emptiness problem can be shown Skolem-hard by encoding the positivity problem which is known to be Skolem-hard, see \cite{OW14} for example. Essentially this stems from the fact that $u_j^2$ is also an LRS and clearly $u_j^2$ is nonnegative. Thus, encoding $u_j^2$ in the same way as above means that $\pfa(a^j) \geq \frac{1}{(k+2)}$ with equality if and only if $u_j = 0$ as required.

Note that stochastic matrix $A$ is necessarily \emph{exponentially ambiguous} so long as one of the coefficients $a_0, \ldots, a_{k-1}$ is negative. This follows since $A$ is then strictly positive by construction. The degree of ambiguity of word $a^j$ is thus $(k+2)^j$. An interesting construction is shown in \cite{OW12}, where a technique to generate a stochastic matrix of dimension $2k+1$ is given (rather than dimension $k+2$ shown above). The resulting PFA is still exponentially ambiguous, but with lower ambiguity. 
\end{proof} 
\renewcommand{\proofname}{\oldproofname}

We now move to prove our main result, specifically that the emptiness and $\lambda$-reachability problems for polynomially ambiguous unary probabilistic finite automata are in EXPTIME. Note again that without the restriction of polynomial ambiguity the problem is Skolem-hard by Theorem~\ref{skolRes} and thus not even known to be decidable.

\begin{thm}\label{mainthm}
The $\lambda$-reachability and (strict) emptiness problems for unary polynomially ambiguous Probabilistic Finite Automata are decidable in EXPTIME. 
\end{thm}

In order to establish Theorem~\ref{mainthm}, we need to prove a series of lemmas.

The next lemma states that we may consider a unary polynomially ambiguous PFA whose transition matrix is upper-triangular. This will prove useful since in that case the eigenvalues of the transition matrix are rational nonnegative. In general, a polynomially ambiguous unary PFA may have a transition matrix with complex eigenvalues as illustrated here: 
\[
A = \begin{pmatrix}
0 & 1 & 0 \\ 0 & 0 & 1 \\ 1 & 0 & 0
\end{pmatrix},
\]
whose eigenvalues are  $\left\{1, -\frac{1}{2} \pm {\bf i}{\frac{\sqrt 3}{2}}\right\}$. The proof of the lemma relies on the analysis of strongly connected components (SCCs) of the underlying transition graph of a PFA.

The following lemma is folklore although we could not find a definite reference which proved the statement and therefore we include a proof for completeness.

\begin{lem}\label{lem:split}
Let $\mathcal{P} = (u, A, v)$ be a polynomially ambiguous unary Probabilistic Finite Automaton with acceptance function $\mathcal{P}(a^k) = u^TA^kv$. Then we can compute in EXPTIME a set of $d$ polynomially ambiguous unary PFAs $\{\mathcal{P}_{s} = (u_s, U, v')\ |\ 0 \leq s \leq d-1\}$ such that $U$ is rational upper-triangular and
$\mathcal{P}(a^k) = \mathcal{P}_{s}(a^{r}) = u_{s}^TU^rv'$, where $k = rd + s$ with $0\leq s \leq d-1$.
\end{lem}

\begin{proof}
We will identify $\mathcal{P}$ and its underlying graph in which an edge $(p,q)$ exists if{}f $A_{p,q}\neq 0$. Two states $p, q$ of a PFA are said to be \emph{connected} if there exists a path from $p$ to $q$ and from $q$ to $p$. We partition the set of states into Strongly Connected Components (SCC) denoted  $S_1, S_2, \ldots, S_{\ell}$ so that for any SCC $S_j$, either $|S_j| = 1$, or else any two states in $S_j$ are connected. These SCCs can be computed in linear time.

A polynomially ambiguous PFA does not have the EDA property (see Sec.~\ref{amb-sec}). This implies that every $S_j$, with $|S_j| > 1$, consists of a single directed cycle, possibly with transitions to other SCCs. To see this, suppose there are two different directed cycles inside $S_j$ of lengths $m$ and $n$ and a common vertex $p$. Then one can construct two different paths of length $mn$ from $p$ to $p$ by going $m$ times along the first cycle and $n$ time along the second cycle, respectively, contradicting the assumption that $\pfa$ does not have the EDA property.

Note that if there exists a path from a state $p \in S_{j_1}$ to some $q \in S_{j_2}$, then there does not exist any path from any state in $S_{j_2}$ to a state in $S_{j_1}$, otherwise $S_{j_1}$ and $S_{j_2}$ would merge to a single SCC (since all vertices are then connected).
This implies that the connected components $S_1, S_2, \ldots, S_{\ell}$ can be reordered in such a way that there are no transitions from $S_j$ to $S_i$ for $i<j$. Hence there exists a permutation matrix $P$ such that the following matrix is stochastic  block upper-triangular:
\[
B = PAP^{-1} = \begin{pmatrix} B_1 & * & \cdots & * \\ 0 & B_2 & \ddots & * \\ \vdots & \ddots & \ddots & \vdots \\ 0 & 0 & \cdots & B_\ell\end{pmatrix},
\]
such that each $B_j \in \Q^{d_j \times d_j}$, where $d_j$ is the size of $S_j$, and $B_j \preceq P_j$, where $P_j \in \N^{d_j \times d_j}$ is a permutation matrix, and the entries $*$ are arbitrary. Here $M\preceq N$ means that $M$ is entrywise less than or equal to $N$, i.e.\ $M_{i,j} \leq N_{i,j}$. 

Let $d = \textrm{lcm}\{d_j\ |\ 1 \leq j \leq \ell\}$ (in fact, we can simply take $d=\prod_{j=1}^\ell d_j$). We then see that:
\[
U := B^d = PA^dP^{-1} = \begin{pmatrix} B_1^d & * & \cdots & * \\ 0 & B_2^d & \ddots & * \\ \vdots & \ddots & \ddots & \vdots \\ 0 & 0 & \cdots & B_\ell^d\end{pmatrix}.
\]
Note that each $B_j^d \preceq P_j^d = I_j$, where $I_j \in \N^{d_j \times d_j}$ is the identity matrix, and the entries~$*$ are arbitrary. Therefore, each $B_j^d$ is diagonal, and so $U$ is clearly upper-triangular. 

We then define $\mathcal{P}_{s} = (u_{s}, U, v')$, for $0 \leq s \leq d-1$, with $u_{s}^T = u^T A^{s} P^{-1}$ and $v' = Pv$ noting that $Pv$ is a binary vector as required of a final state vector. We now see that:
\[
\mathcal{P}(a^k) = u^TA^kv = u^TA^sA^{rd}v = u^TA^sP^{-1}(PA^{rd}P^{-1})Pv  = u_{s}^TU^rv' = \mathcal{P}_{s}(a^r) 
\]
for $k = rd + s$ with $0\leq s \leq d-1$ as required. Here we used the identity $U^r = PA^{rd}P^{-1}$.

Finally, note that $d$ can be exponential in the number of states of $\mathcal{P}$, which in turn is bounded by the input size. Hence computing $U$ and all $u_{s}$, for $0 \leq s \leq d-1$, takes exponential time. 
\end{proof}

The next lemma gives us an efficient method to compute an explicit formula for the acceptance probability function of a unary PFA with upper-triangular transition matrix.\footnote{A similar result is derived in Proposition 1 of \cite{ABM20}, without the claim of computing the polynomials in polynomial time.}

\begin{lem}\label{toEqnLem}
Let $\mathcal{P} = (u, A, v)$ be a unary probabilistic finite automaton such that $A$ is rational upper-triangular, and let $\lambda_0 = 1 > \lambda_1 > \cdots >\lambda_m\geq 0$ be distinct eigenvalues of $A$. Then there exist a constant $c \in \Q$ and univariate polynomials $p_1, \ldots, p_{m}$ over $\Q$, all of which can be computed in polynomial time, such that
\[
\mathcal{P}(a^k) = c + \sum_{i=1}^{m} p_i(k) \lambda_i^{k}.
\]
\end{lem}

\begin{proof}
First, we write $A$ in \emph{Jordan normal form} $A = S^{-1} J S$, where $S$ is a nonsingular (det$(S) \neq 0$) matrix consisting of the generalised eigenvectors of $A$. Recall that $A$ is a rational upper-triangular matrix. It follows that $J$ and $S$ must have rational entries because to compute them, we need to solve systems of linear equations over $\Q$. These calculations, that is, computing $J$, $S$ and $S^{-1}$ can be done in polynomial time. In fact, these problems are in NC, see \cite{BorodinGathenHopcroft82,Pap94}. Matrix $J$ has the form $J= \bigoplus_{i=0}^m \bigoplus_{j=1}^{n_i} J_{\ell_{i,j}}(\lambda_{i})$, where $J_{\ell_{i,j}}(\lambda_{i})$ is a $\ell_{i,j} \times \ell_{i,j}$ \emph{Jordan block} and $n_i$ is the geometric multiplicity of $\lambda_i$ (hence $\sum_{j=1}^{n_i} \ell_{i,j}$ is the algebraic multiplicity of $\lambda_i$). Recall that a Jordan block $J_{\ell}(\lambda)$ of size $\ell \times \ell$  that corresponds to an eigenvalue $\lambda$ has the form:
\[
J_{\ell}(\lambda) = \begin{pmatrix}\lambda & 1 & 0 & \cdots & 0 \\ 0 & \lambda & 1 & \cdots & 0 \\ 0 & 0 & \lambda & \cdots & 0 \\ \vdots & \vdots & \vdots & \ddots & \vdots \\ 0 & 0 & 0 & \cdots & \lambda \end{pmatrix} \in \mathbb{Q}^{\ell \times \ell}.
\]
Noting that $\binom{x}{y} = 0$ if $y > x$, we see that
\begin{equation}
J_{\ell}(\lambda)^{k} = \begin{pmatrix}\lambda^{k} & \binom{k}{1}\lambda^{k-1} & \binom{k}{2}\lambda^{k-2} & \cdots & \binom{k}{\ell-1}\lambda^{k-(\ell-1)} \\ 0 & \lambda^{k} & \binom{k}{1}\lambda^{k-1} & \cdots & \binom{k}{\ell-2}\lambda^{k-(\ell-2)} \\ 0 & 0 & \lambda^{k} & \cdots & \binom{k}{\ell-3}\lambda^{k-(\ell-3)} \\ \vdots & \vdots & \vdots & \ddots & \vdots \\ 0 & 0 & 0 & \cdots & \lambda^{k} \end{pmatrix}.
\label{jordanform}
\end{equation}
Note that the entries of $J_{\ell}(\lambda)^{k}$ have the form $q_{i,j}(k)\lambda^k$, where $q_{i,j}(k)$ are polynomials over $\Q$ that can be computed in polynomial time. Namely, $q_{i,i+p}(k) = \binom{k}{p}\lambda^{-p}$ for $0 \leq p\leq \ell-i$, and $q_{i,j}(k) = 0$ for $i>j$. Note that even though $p$ appears in the exponent of $\lambda^{-p}$ and as $p!$ in $\binom{k}{p}$, these values are still computable in PTIME from the input data because $p$ is bounded by the dimension of the matrix, which in turn is bounded by the input size.

Next, we note that $J^k = \bigoplus_{i=0}^m \bigoplus_{j=1}^{n_i} J_{\ell_{i,j}}(\lambda_{i})^k$. Hence the entries of $J^k$ have the form $p_{s,t}(k)\lambda_i^k$, where $p_{s,t}(k)$ are polynomials over $\Q$. So we can write the function $\mathcal{P}(a^k)$ as follows:
\[
\mathcal{P}(a^k) = u^T A^k v = (u^T S^{-1}) J^k (Sv).
\]
Note that in the above equation, $u^T S^{-1}$ and $Sv$ are rational vectors. It follows that
\[
\mathcal{P}(a^k) = \sum_{i=0}^{m} p_i(k) \lambda_i^{k}
\]
for some polynomials $p_i(k)$ over $\Q$. In fact, these polynomials are rational linear combinations of those $p_{s,t}(k)$ that multiply $\lambda_i^k$ in the expression for $J^k$, and so they can be computed in polynomial time.

Finally, recall that $\lambda_0=1$ and note that the Jordan blocks that correspond to the dominant eigenvalues of a stochastic matrix have size $1\times 1$ (for the proof of this fact see, e.g.\ \cite[Theorem~6.5.3]{fried}). It follows from (\ref{jordanform}) that the terms $\lambda_0^k$ in the formula for $J^k$ are multiplied by constant polynomials $p_{s,t}(k)=1$. Hence $p_0(k)=c$ for some constant $c\in \Q$. 
\end{proof}

The next technical lemma is crucial in our later analysis of the running time of the algorithms for the emptiness and $\lambda$-reachability problems presented in Lemmas \ref{lem:case1} and \ref{lem:case2}.

\begin{lem}\label{lem:log}
Let $D\in \R$ be such that $\ln D>2$. Then for all $x>3D\ln D$, we have $D \ln x < x$.
\end{lem}

\begin{proof}
Our goal is to find $x_0>0$ such that every $x>x_0$ satisfies $D \ln x < x$. First, let us make a substitution $x=Dt$, where $t>1$. Then we can rewrite $D \ln x < x$ as follows
\begin{align*}
    D \ln(Dt) &< Dt,\\
    \ln t + \ln D &< t.
\end{align*}
We want to find $t_0>1$ such that every $t>t_0$ satisfies $\ln t + \ln D < t$. Let us make another substitution $t=\ln D + u \ln \ln D$, where $u>0$. Then we can write the previous inequality as
\begin{align}
\ln(\ln D + u \ln \ln D) + \ln D &< \ln D + u \ln \ln D, \nonumber \\
\ln\left(\ln D \left(1 + u \frac{\ln \ln D}{\ln D}\right)\right) &< u \ln \ln D, \nonumber \\
\ln\ln D  + \ln \left(1 + u \frac{\ln \ln D}{\ln D}\right) &< u \ln \ln D. \label{eq:c}
\end{align}
So we need to find $u_0>0$ such that for all $u>u_0$, the inequality (\ref{eq:c}) holds. In order to do this, we can replace the left-hand side of (\ref{eq:c}) with a larger value using $\ln \left(1 + u \frac{\ln \ln R}{\ln R}\right) < u \frac{\ln \ln R}{\ln R}$. Thus we obtain
\begin{align*}
\ln\ln D  + u \frac{\ln \ln D}{\ln D} &< u \ln \ln D,\\
1 + \frac{u}{\ln D} &< u, \qquad \ln D +u < u\ln D, \qquad \frac{\ln D}{\ln D-1} < u.
\end{align*}
Recall that by our assumption $\ln D>2$. In this case, $\frac{\ln D}{\ln D-1} < 2$, and hence we can choose $u_0=2$. This gives us the values $t_0 = \ln D + u_0 \ln \ln D = \ln D + 2 \ln \ln D$ and $x_0 = D t_0 = D ( \ln D + 2 \ln \ln D)$. Since $\ln \ln D < \ln D$, we can choose $x_0$ to be $x_0=3D \ln D$.
\end{proof}

We now proceed to the proof of our main result. We split the analysis into two cases depending on whether or not the cutpoint $\lambda$ coincides with the limit $\lim_{k \to \infty}\mathcal{P}(a^k)$, which is unique by Lemma~\ref{toEqnLem}.

\begin{lem}\label{lem:case1bound}
Let $\mathcal{P} = (u, A, v)$ be a unary probabilistic finite automaton such that $A$ is rational upper-triangular. Let $\lambda \in [0,1] \cap \Q$ be a cutpoint and assume that $c = \lim_{k \to \infty}\mathcal{P}(a^k)\neq \lambda$. Then we can compute in PTIME a bound $k_0>0$ such that for all $k> k_0$, we have
\[
|\mathcal{P}(a^k) - c| < \frac{|c-\lambda|}{2}.
\]
Moreover, we give an explicit formula for $k_0$ in (\ref{eq:k0}).
\end{lem}

\begin{proof}
By Lemma~\ref{toEqnLem}, we can write
$
\mathcal{P}(a^k) = c + \sum_{i=1}^{m} p_i(k) \lambda_i^{k}
$,
where $1>\lambda_1> \cdots >\lambda_m$ are the eigenvalues of $A$. Note that $c$, $\lambda_i$'s and the coefficients of $p_i$ are rational numbers that can be computed in polynomial time. By assumption, $\lim_{k \to \infty} \mathcal{P}(a^k) = c \neq \lambda$. Let $\epsilon = |c-\lambda|/2$. We now show how to compute an integer $k_0 >0$ such that $|\mathcal{P}(a^k) - c| < \epsilon$ for all $k > k_0$.

Let each $p_i(k)$ have the form $p_i(k) = a_{i,s} k^s +a_{i,s-1}k^{s-1} + \cdots + a_{i,0}$, where $s\leq n$ is the size of the largest Jordan block in the Jordan normal form of $A$ (we do not assume here that $a_{i,s}\neq 0$). Then for all $k>0$ we have
\[
\left| \sum_{i=1}^{m} p_i(k) \lambda_i^{k} \right| \leq \lambda_1^k \sum_{i=1}^{m} \left| p_i(k) \right| \leq \lambda_1^k k^s \sum_{i=1}^{m} \sum_{j=0}^{s} \left| a_{i,j} \right| = d\, k^s \lambda_1^k,
\]
where $d = \sum_{i=1}^{m} \sum_{j=0}^{s} \left| a_{i,j} \right|\in \Q$
can be computed in polynomial time by Lemma \ref{toEqnLem}.

Let $k_1>0$ be a number to be defined later such that for all $k>k_1$,
\[
k^s < \left( \frac{1}{\sqrt{\lambda_1}} \right)^k = \lambda_1^{-\frac{k}{2}}.
\]
Then for all $k>k_1$, we have $d\, k^s \lambda_1^k < d \lambda_1^{\frac{k}{2}}$. Thus we need to find $k_0\geq k_1$ such that for all $k>k_0$, we have $\lambda_1^{\frac{k}{2}} < \epsilon/d$. Note that if $\epsilon/d\geq 1$, then we can take $k_0=k_1$. Hence we assume that $\epsilon/d<1$.

The inequality $\lambda_1^{\frac{k}{2}} < \epsilon/d$ is equivalent to $k \ln \lambda_1 < 2\ln(\epsilon/d)$. Since $\ln \lambda_1<0$, the previous inequality is equivalent to
\begin{equation}\label{eq:a}
k> \frac{2\ln(\epsilon/d)}{\ln \lambda_1} = \frac{2\ln(d/\epsilon)}{-\ln \lambda_1}.
\end{equation}
To determine $k_0$, we need an upper bound on the right-hand side of (\ref{eq:a}). We will use the fact that for any rational $r>1$, $
\ln r < \log_2 r \leq \log_2\lceil r \rceil < \mathrm{bins}(\lceil r \rceil)$,
where $\mathrm{bins}(n)$ is the size of the binary representation of $n$. Thus $\mathrm{bins}(\lceil r \rceil)$ gives a polynomially computable integer upper bound for $\ln r$.

Next, using the fact that $\ln(1+x) < x$ for $x\neq 0$, we obtain
\[
\ln \lambda_1  = \ln (1+(\lambda_1-1)) < \lambda_1-1,
\]
which gives $-\ln \lambda_1 > 1-\lambda_1$. Hence a polynomially computable upper bound on the right-hand side of (\ref{eq:a}) is
\begin{equation}\label{eq:d}
     \frac{2\ln(d/\epsilon)}{-\ln \lambda_1} < \frac{2\, \mathrm{bins}(\lceil d/\epsilon \rceil)}{1-\lambda_1}.
\end{equation}

Next we compute a value $k_1$ such that for all $k>k_1$:
\begin{equation}\label{eq:b}
    k^s < \lambda_1^{-\frac{k}{2}}\qquad \text{or, equivalently,}\qquad C \ln k < k,
\end{equation}
where $C=\dfrac{2s}{-\ln \lambda_1}$.
Using the fact that $\ln(1+x)<x$ for $x\neq 0$, we obtain $C<\dfrac{2s}{1-\lambda_1}$. Hence in order to find $k_1$, we can replace $C$ in (\ref{eq:b}) with $D = \dfrac{2s}{1-\lambda_1}$. In addition, we can assume that $\ln D>2$, since otherwise we can replace $D$ with a larger value that satisfies this condition, e.g.\ with $D=9$. Now, Lemma \ref{lem:log} implies that every $k>3D\ln D$ satisfies $D\ln k < k$. To make this value polynomially computable, we can choose it to be
\[
k_1 = 3 \lceil D \rceil \mathrm{bins}(\lceil D \rceil), \text{ where } D = \max \left\{ \dfrac{2s}{1-\lambda_1},\, 9\right \}.
\]
Finally, combining the right-hand side of (\ref{eq:d}) with the above formula, we can define
\begin{equation}\label{eq:k0}
k_0 = \max \left\{\dfrac{2\, \mathrm{bins}(\lceil d/\epsilon \rceil)}{1-\lambda_1},\   3 \lceil D \rceil \mathrm{bins}(\lceil D \rceil)\right\}.
\end{equation}
Note that all the values that appear in the above formula, e.g.\ $\epsilon$, $d$ and $D$, can be computed in polynomial time from the input data.
\end{proof}

\begin{lem}\label{lem:case1}
Let $\mathcal{P} = (u, A, v)$ be a unary probabilistic finite automaton such that $A$ is rational upper-triangular, and let $\lambda \in [0,1] \cap \Q$ be a cutpoint. Assuming that $\lambda \neq \lim_{k \to \infty}\mathcal{P}(a^k)$, the (strict) emptiness and $\lambda$-reachability problems for $\mathcal{P}$ and $\lambda$ are decidable in EXPTIME.
\end{lem}

\begin{proof}
In Lemma \ref{lem:case1bound}, we have derived a polynomially computable bound $k_0$ such that $\mathcal{P}(a^k) = u^T A^k v \in (c - \epsilon, c+ \epsilon)$, where $\epsilon = |c-\lambda|/2$. In particular, $\mathcal{P}(a^k) \neq \lambda$ for all $k > k_0$. Now, to decide the $\lambda$-reachability problem, we need to check for each integer $k \in [0, k_0]$ whether $u^T A^k v = \lambda$. Note that the number of integers in $[0, k_0]$ is $\mathcal{O}(2^{\mathrm{bins}(k_0)})$, which is exponential in the instance size. Also, computing $A^k$ for a given $k \in [0, k_0]$ takes exponential time because $\mathrm{bins}(A^k) = \mathcal{O}(2^{\mathrm{bins}(k_0)}\mathrm{bins}(A))$. So, the overall algorithm is in EXPTIME.

In a similar way, we can decide the (strict) emptiness problem in EXPTIME. For instance, suppose $\lambda >c$. Then for all $k>k_0$, we have $\mathcal{P}(a^k) < c+\epsilon < \lambda$. Thus deciding whether there exists $k$ such that $\mathcal{P}(a^k) < \lambda$ is trivial. Suppose we want to know if there exists $k$ such that $\mathcal{P}(a^k) \geq \lambda$. In this case, we need to check for each integer $k \in [0, k_0]$ whether $u^T A^k v \geq \lambda$. By the same argument as before, this can be done in EXPTIME.
\end{proof}

\begin{lem}\label{lem:case2bound}
Let $\mathcal{P} = (u, A, v)$ be a unary polynomially ambiguous probabilistic finite automaton such that $A$ is upper-triangular. Let $\lambda \in [0,1] \cap \Q$ be a cutpoint such that $\lambda = \lim_{k \to \infty}\mathcal{P}(a^k)$. Then we can compute in PTIME a bound $k_0>0$ such that
\begin{align*}
\text{either}\quad \mathcal{P}(a^k) &> \lambda\quad \text{for all } k>k_0,\\
\text{or}\quad \mathcal{P}(a^k) &< \lambda\quad \text{for all } k>k_0.
\end{align*}
Moreover, we give an explicit formula for $k_0$. Namely, $k_0 = \max\{k_1, k_2\}$, where $k_1$ is defined in (\ref{eq:x}) and $k_2$ is defined in (\ref{eq:k2_1}) and (\ref{eq:k2_2}).
\end{lem}

\begin{proof}
Recall that by Lemma~\ref{toEqnLem}, we can write
$
\mathcal{P}(a^k) = c + \sum_{i=1}^{m} p_i(k) \lambda_i^{k}
$,
where $1>\lambda_1> \cdots >\lambda_m$ are the eigenvalues of $A$, and $c$, $\lambda_i$'s and the coefficients of $p_i$ are rational numbers that can be computed in polynomial time. By our assumption, $\lambda = \lim_{k \to \infty}\mathcal{P}(a^k) = c$. Let each $p_i(k)$ have the form $p_i(k) = a_{i,s} k^s +a_{i,s-1}k^{s-1} + \cdots + a_{i,0}$, where $s\leq n$ is the size of the largest Jordan block in the Jordan normal form of $A$ (we do not assume here that $a_{i,s}\neq 0$).

In addition, assume that the leading coefficient of $p_1(k)$ is $a_{1,t}$, for some $t\leq s$. Suppose that $a_{1,t}>0$. We now show how to compute $k_0$ such that $\mathcal{P}(a^k) > \lambda$ for all $k>k_0$. (The case when $a_{1,t}<0$ and $\mathcal{P}(a^k) < \lambda$ for all $k>k_0$ is symmetric).

First, we compute $k_1$ such that $p_1(k)> \frac{1}{2}a_{1,t} k^t$ for all $k > k_1$. To do this, we will use the following inequalities:
\begin{align*}
    &a_{1,t} k^t +a_{1,t-1}k^{t-1} + \cdots + a_{1,0} > \frac{1}{2}a_{1,t} k^t\ \ \Longleftrightarrow\ \ \frac{1}{2}a_{1,t} k^t +a_{1,t-1}k^{t-1} + \cdots + a_{1,0} > 0\\
    &\text{and } |a_{1,t-1}k^{t-1} + \cdots + a_{1,0}| \leq k^{t-1}(|a_{1,t-1}| + \cdots + |a_{1,0}|) = k^{t-1}\sum_{j=0}^{t-1} |a_{1,j}| \quad \text{if } k\geq 1.
\end{align*}
So, the inequality $p_1(k)> \frac{1}{2}a_{1,t} k^t$ follows from $\frac{1}{2}a_{1,t} k^t > k^{t-1}\sum_{j=0}^{t-1} |a_{1,j}|$, which is equivalent to $k > \frac{2}{a_{1,t}}\sum_{j=0}^{t-1} |a_{1,j}|$. Therefore, we conclude that
\begin{equation}\label{eq:x}
p_1(k)> \frac{1}{2}a_{1,t} k^t\quad \text{for all $k$ such that } k > k_1 := \max \left\{1,\ \frac{2}{a_{1,t}} \sum_{j=0}^{t-1} |a_{1,j}|\right\}.
\end{equation}

Now we want to find $k_2\geq k_1$ such that for all $k>k_2$, we have
\begin{equation}\label{eq:y}
\lambda_1^k p_1(k) + \lambda_2^k p_2(k) + \cdots + \lambda_m^k p_m(k) > 0.
\end{equation}
Note that
\begin{equation}\label{eq:z}
|\lambda_2^k p_2(k) + \cdots + \lambda_m^k p_m(k)| \leq \lambda_2^k (|p_2(k)| + \cdots + |p_m(k)|) \leq d k^s \lambda_2^k,
\end{equation}
where $d=\sum_{i=2}^m \sum_{j=0}^s |a_{i.j}|$. Using (\ref{eq:x}) and (\ref{eq:z}), we see that (\ref{eq:y}) holds whenever $k>k_1$ and $dk^s \lambda_2^k < \frac{1}{2}a_{1,t} k^t \lambda_1^k$, which is equivalent to
\begin{align}
\frac{2dk^{s-t}}{a_{1,t}} < {\left(\frac{\lambda_1}{\lambda_2}\right)}^k \qquad \text{or}\qquad
&\ln \frac{2d}{a_{1,t}} + (s-t) \ln k  < k \ln \frac{\lambda_1}{\lambda_2} \nonumber\\
&\frac{1}{\ln \lambda_1/\lambda_2}\left(\ln \frac{2d}{a_{1,t}} + (s-t) \ln k\right)  < k.\label{eq:e}
\end{align}
We will use the following inequality
\[
\ln \frac{\lambda_1}{\lambda_2} = - \ln \frac{\lambda_2}{\lambda_1} = - \ln \left( 1 + \frac{\lambda_2-\lambda_1}{\lambda_1} \right) >
- \frac{\lambda_2-\lambda_1}{\lambda_1} > \lambda_1-\lambda_2.
\]
Then we can replace (\ref{eq:e}) with a stronger inequality
\begin{equation}\label{eq:f}
    \frac{1}{\lambda_1 - \lambda_2}\left(\ln \frac{2d}{a_{1,t}} + (s-t) \ln k\right)  < k.
\end{equation}
In the following, we will assume $t<s$ since otherwise (\ref{eq:f}) simplifies to $\frac{1}{\lambda_1 - \lambda_2}\ln \frac{2d}{a_{1,t}} < k$.
Let us make the substitution $k=t {\left( \frac{2d}{a_{1,t}} \right)}^{-\frac{1}{s-t}}$, where $t>0$. Then (\ref{eq:f}) can be written as
\begin{align*}
\frac{1}{\lambda_1 - \lambda_2}\left(\ln \frac{2d}{a_{1,t}} + (s-t) \ln t + (s-t) \frac{-1}{s-t} \ln \frac{2d}{a_{1,t}}\right)  &< t {\left( \frac{2d}{a_{1,t}} \right)}^{-\frac{1}{s-t}}\\
{\left( \frac{2d}{a_{1,t}} \right)}^{\frac{1}{s-t}}\frac{s-t}{\lambda_1 - \lambda_2} \ln t &< t.
\end{align*}
Let $D = \max \left\{9,\ {\left( \frac{2d}{a_{1,t}} \right)}^{\frac{1}{s-t}}\frac{s-t}{\lambda_1 - \lambda_2}\right\}$. Here $9$ is needed to satisfy the requirement $\ln D>2$ in Lemma \ref{lem:log}. Then by Lemma \ref{lem:log}, the above inequality holds when $t> 3D \ln D$. Therefore, (\ref{eq:f}) and hence (\ref{eq:e}) holds when $k > 3{\left( \frac{2d}{a_{1,t}} \right)}^{-\frac{1}{s-t}}D \ln D$. To make this bound polynomially computable, we can simplify it as follows. Suppose that $2d \geq a_{1,t}$. Then (\ref{eq:e}) holds when
\begin{equation}\label{eq:k2_1}
k > k_2 := 3\lceil E \rceil \mathrm{bins}(\lceil E \rceil), \quad \text{where } E = \max \left\{9,\ \frac{2d}{a_{1,t}}\cdot \frac{s-t}{\lambda_1 - \lambda_2}\right\}
\end{equation}
because in this case ${\left( \frac{2d}{a_{1,t}} \right)}^{\frac{1}{s-t}} \leq \frac{2d}{a_{1,t}}$ and ${\left( \frac{2d}{a_{1,t}} \right)}^{-\frac{1}{s-t}} \leq 1$. On the other hand, if $2d < a_{1,t}$, then (\ref{eq:e}) holds when
\begin{equation}\label{eq:k2_2}
k > k_2 := 3 \left\lceil\frac{a_{1,t}}{2d}E\right\rceil \mathrm{bins}(\lceil E \rceil), \quad \text{where } E = \max \left\{9,\ \frac{s-t}{\lambda_1 - \lambda_2}\right\}
\end{equation}
because in this case ${\left( \frac{2d}{a_{1,t}} \right)}^{-\frac{1}{s-t}} < {\left(\frac{2d}{a_{1,t}}\right)}^{-1}$ and ${\left( \frac{2d}{a_{1,t}} \right)}^{\frac{1}{s-t}} < 1$.

Finally, we conclude that (\ref{eq:y}) holds for all $k> k_0 :=\max \{k_1,k_2\}$, where both $k_1$ and $k_2$ are computable in PTIME.
\end{proof}

\begin{lem}\label{lem:case2}
Let $\mathcal{P} = (u, A, v)$ be a unary polynomially ambiguous probabilistic finite automaton such that $A$ is upper-triangular and let $\lambda \in [0,1] \cap \Q$ be a cutpoint. Assuming that $\lambda = \lim_{k \to \infty}\mathcal{P}(a^k)$, the (strict) emptiness and $\lambda$-reachability problems for $\mathcal{P}$ and $\lambda$ are decidable in EXPTIME.  
\end{lem}

\begin{proof}
In Lemma \ref{lem:case2bound}, we obtained a polynomially computable value $k_0$ such that either $\mathcal{P}(a^k) > \lambda$ for all $k>k_0$ or $\mathcal{P}(a^k) < \lambda$ for all $k>k_0$. Using the same argument as at the end of the proof of Lemma \ref{lem:case1}, we can show that the (strict) emptiness and $\lambda$-reachability problems are decidable in EXPTIME.
\end{proof}

We are now ready to give a proof of Theorem~\ref{mainthm}.

\begin{proof}[Proof of Theorem \ref{mainthm}]
Let $\mathcal{P} = (u, A, v)$ be a polynomially ambiguous unary PFA. By Lemma~\ref{lem:split}, we can compute in EXPTIME a set of $d$ polynomially ambiguous unary PFAs $\{\mathcal{P}_{s} = (u_s, U, v')\ |\ 0 \leq s \leq d-1\}$ such that $U$ is rational upper-triangular and
\[
\mathcal{P}(a^{rd+s}) = \mathcal{P}_{s}(a^{r}) = u_{s}^TU^rv',
\]
where $0\leq s \leq d-1$.

Suppose $\lambda$ is a given cutpoint. If we want to decide whether there exists $k$ such that $\mathcal{P}(a^{k}) = \lambda$ (or $\mathcal{P}(a^{k}) \geq \lambda$), we can check for every $s$ from $0$ to $d-1$ whether there exists $r$ such that $\mathcal{P}_s(a^{r}) = \lambda$ (or $\mathcal{P}_s(a^{r}) \geq \lambda$, respectively), which can be done in EXPTIME using Lemmas \ref{lem:case1} and \ref{lem:case2}. Namely, we will use Lemma \ref{lem:case1} if $\lambda \neq c_s$ and Lemma \ref{lem:case2} if $\lambda = c_s$ for the current values of $s\in [0,d-1]$. Finally, we note that even though the value of $d$ can be exponential in the input size, the whole procedure can still be done in EXPTIME.
\end{proof}

Skolem's problem is at least NP-hard \cite{BP02} implying that the $\lambda$-reachability and emptiness problems are also NP-hard, at least for PFA of exponential ambiguity. Our next result shows that NP-hardness can be established even for unary PFAs of \emph{finite ambiguity}.

\begin{prop}\label{thm:NPhard}
The $\lambda$-reachability and emptiness problems for unary finitely ambiguous Probabilistic Finite Automata $\mathcal{P} = (u, A, v)$ with $\{0,1\}$-matrix $A$ are NP-hard.
\end{prop}

\begin{proof}
The NP-hardness of Skolem's problem was established in \cite{BP02}. Specifically, Corollary~1.3 of \cite{BP02} states that the problem of determining, for a given matrix $A \in \{0, 1\}^{n \times n}$ and row vectors $b, c \in \{0, 1\}^n$, if $b^TA^kc = 0$ for some $k \geq 0$ is NP-hard. 
Examination of the proof of this corollary shows that in fact $\mathcal{P}$ is finitely ambiguous as we shall show.

The proof of Theorem~1.1 of \cite{BP02} shows a reduction of 3SAT on $m$ clauses with $n$ letters to a unary rational expression $E$ of the form:
\[
E = \bigcup_{j = 0}^{k} a^{z_j}(a^{r_j})^* ,
\]
where $k = \mathcal{O}(n^3 m)$ and $z_j, r_j = \mathcal{O}(n^6)$ as is not difficult to see from the proof in \cite{BP02}. Notice then that each $z_j, r_j$ represented in unary has a polynomial size in terms of the 3SAT instance and thus $E$ also has a polynomial representation size.

We may then invoke Kleene's theorem \cite{Kl56} to state that the language recognised by $E$ is also recognised by an NFA $\mathcal{P} = (b, \{A\}, c)$ which thus allows the derivation of  Corollary~1.3 of \cite{BP02}. Note that $E$ is simply the union of rational expressions of the form $E_j = a^{z_j}(a^{r_j})^*$. Each $E_j$ can be transformed to an NFA $\mathcal{N}_j$ with $z_j+r_j+1$ states $S_j = n_{0,j}, \ldots, n_{z_j,j}, n_{z_j+1,j}, \ldots, n_{z_j+r_j,j}$ with initial state $n_{0,j}$, final state $n_{z_j+1,j}$ and transition function $\delta: S_j \times \{a\} \to S_j$ given by $\delta(n_{i,j},a) = n_{i+1,j}$ for $0 \leq i \leq z_j+r_j-1$ and $\delta(n_{z_j+r_j},a) = n_{z_j+1,j}$. 

We may then form an NFA $\mathcal{N}$ by $\mathcal{N} = \bigcup_{j=0}^k \mathcal{N}_j$ with the usual construction. In this case, $\mathcal{N}$ has set of initial states $\{n_{0,j}\ |\ 1 \leq j \leq k\}$, set of final states $\{n_{z_j+1,j}\ |\ 1 \leq j \leq k\}$ and states in disjoint subsets $S_j$ and $S_{j'}$ with $j \neq j'$ are not connected. This implies by the IDA property of \cite{WS91} that $\mathcal{N}$ is finitely ambiguous since there does not exist any state with two outgoing transitions (by which reasoning we also know that each row of $\mathcal{N}$'s transition matrix has exactly one entry $1$ with all others $0$). In fact one may see that $\mathcal{N}$ is $k$-ambiguous with $k = \mathcal{O}(n^3 m)$. The number of states of $\mathcal{N}$ is $d = \sum_{j = 0}^{k}z_j+r_j+1 = \mathcal{O}(n^9m)$ which is polynomial in the 3SAT instance representation size.

We note that actually $\mathcal{N}$ is already close to a PFA. Since each row  is zero except for exactly one entry $1$, matrix $A$ is stochastic. We thus consider Probabilistic Finite Automaton $\mathcal{P} = (u, \{A\}, c)$ where $u = \frac{b}{|b|}$ is the initial (stochastic) vector. $\mathcal{P}$  has polynomial ambiguity since $\mathcal{N}$ does. Therefore, deciding if there exists $k \geq 0$ such that $\mathcal{P}(a^k) = 0$ or $\mathcal{P}(a^k) \leq 0$ is NP-hard to determine, proving NP-hardness of the $\lambda$-reachability and emptiness problems. Since we did not modify $\mathcal{N}$ to derive $\mathcal{P}$ other than to scale the initial vector, the degree of ambiguity is retained.
%We now consider the Probabilistic Finite Automaton $\mathcal{P} = (u, \{B\}, v)$ where $u = \frac{[b | 0]}{|b|} \in \mathbb{Q}^{d+1}$, $v = [c | 0] \in \mathbb{Q}^{d+1}$ and
%\[
%B = \frac{1}{n}\begin{pmatrix} A & | & {\bf (n-1)}_n \\ \hline {\bf 0} & | & n\end{pmatrix} \in \mathbb{Q}^{(d+1) \times (d+1)},
%\] 
%with ${\bf (n-1)} = (n-1, \ldots, n-1) \in \mathbb{N}^d$.   
\end{proof}

\begin{cor}\label{cor:NPcmpl}
The $\lambda$-reachability and emptiness problems for unary polynomially ambiguous PFA $\mathcal{P} = (u, A, v)$ with $\{0,1\}$-matrix $A$ are NP-complete.
\end{cor}

\begin{proof}
NP-hardness follows from Proposition~\ref{thm:NPhard} since finite ambiguity is a stronger property than polynomially ambiguity. To prove the NP upper bound, we will show that the algorithm in the proof of Theorem \ref{mainthm} can be done in NP. We again use Lemmas \ref{lem:split}, \ref{toEqnLem}, \ref{lem:case1} and \ref{lem:case2}. Note that the value $d$ from Lemma \ref{lem:split} can be exponential. However, its binary presentation has polynomial size. So, instead of cycling though all $s$ from $0$ to $d-1$, we can nondeterministically guess in polynomial time a value $s\in [0,d-1]$.

Next, we note that the values of $k_0$ in Lemma \ref{lem:case1} and $k_2$ in Lemma \ref{lem:case2} also have binary representations of polynomial size. Again, instead of checking every $k$ in $[0,k_0]$ or $[0,k_2]$, we can nondeterministically guess $k$ in polynomial time.

Finally, in the verification step of our algorithm we need to compute the matrices $A^d$, $A^s$ and $(A^d)^k$. This can be done in polynomial time using exponentiation by squaring. Indeed, the exponentiation by squaring requires polynomially many steps. Also, any power of a stochastic $\{0,1\}$-matrix is also a stochastic $\{0,1\}$-matrix, so the entries of the power matrices do not grow in size.
\end{proof}

As an additional application of the techniques developed within this section, we can also show that the value of a unary PFA can be determined.

\begin{cor}
Let $\mathcal{P} = (u, A, v)$ be a polynomially ambiguous unary Probabilistic Finite Automaton with acceptance function $\mathcal{P}(a^k) = u^TA^kv$. Then we can compute the value of $\mathcal{P}$ in EXPTIME.
\end{cor}

\begin{proof}
By Lemma~\ref{lem:split}, we can compute in EXPTIME a set of $d$ PFA $\{\mathcal{P}_{s} = (u_s, U, v')\ |\ 0 \leq s \leq d-1\}$ which are unary and polynomially ambiguous, such that $U$ is rational upper-triangular and
$\mathcal{P}(a^k) = \mathcal{P}_{s}(a^{r}) = u_{s}^TU^rv'$, where $k = rd + s$ with $0\leq s \leq d-1$. The size of $d$ depends on the structure of $\mathcal{P}$ but is no more than exponential in the description size of $\mathcal{P}$. The value of $\mathcal{P}$ defined by $\sup_{k \geq 0}\mathcal{P}(a^k)$ is clearly the same as $\max_{0 \leq s \leq d}\{\sup_{k \geq 0}\mathcal{P}_s(a^k)\}$ and thus we may simply compute the value of each of these $d$ unary PFA with upper triangular stochastic transition matrices. Assume then for the rest of the proof that $\mathcal{P}$ is a PFA with a rational upper triangular transition matrix $A \in \mathbb{Q}^{n \times n}$.

Lemma~\ref{toEqnLem} shows that we can write the acceptance probability function of $\mathcal{P}$ as:
\[
\mathcal{P}(a^k) = c + \sum_{i=1}^{m} p_i(k) \lambda_i^{k},
\]
where $|\lambda_i| < 1$ and $p_i$ are polynomials.
There are two cases that we consider. Either 
$\sup_{k \geq 0}\mathcal{P}(a^k) = c$, or else $\sup_{k \geq 0}\mathcal{P}(a^k) > c$. 

By defining $\lambda = c$ and then applying Lemma~\ref{lem:case2}, we can determine if the first or second cases holds, i.e., we can determine if there exists a word $a^k$ such that $\pfa(a^k) > \lambda$. If $\pfa(a^\ell) \leq \lambda$ for all $\ell \geq 0$, then the value of $\pfa$ is $\lambda$ (which is reached in the limit) and we may thus stop. Assume then that $\pfa(a^k) > c$ for at least one $k > 0$.

The outline of our approach is now as follows. Suppose the leading coefficient $a_{1,t}$ of the polynomial $p_1(k)$ is positive, that is, $a_{1,t}>0$ (the case when $a_{1,t}<0$ is similar and even simpler). Also, consider the derivative of $\mathcal{P}(a^k)$ with respect to $k$:
\[
\mathcal{P'}(a^k) = \sum_{i=1}^{m} q_i(k) \lambda_i^{k}.
\]
The leading coefficient of $q_1(k)$ is equal to $a_{1,t}\ln \lambda_1$. Note that since $\ln \lambda_1<0$, we have $a_{1,t}\ln \lambda_1<0$.

Using the same ideas as in the proof of Lemma~\ref{lem:case2}, we can compute in polynomial time a value $k_0$ such that $\pfa(a^k) > c$ and $\mathcal{P'}(a^k)<0$ for all $k>k_0$. Notice that the logarithms $\ln \lambda_i$, for $i=1,\ldots,m$, that appear in $\mathcal{P'}(a^k)$ may be irrational. In order to do the above computations in polynomial time, we can use the following rational upper bounds: $\ln \lambda_1 \leq \lambda_1 -1 < 0$ and $|\ln \lambda_i| = \ln (1/\lambda_i) \leq 1/\lambda_i -1$ for $i=1,\ldots,m$. Hence $\pfa(a^k)$ is greater than $c$ and is strictly decreasing for all $k>k_0$.
Therefore, we only need to evaluate $\{\pfa(a^\ell)\ |\ 1 \leq \ell \leq  k_0 \}$ in order to find the value of $\mathcal{P}$. This can be done in EXPTIME using the same argument as at the end of the proof of Lemma~\ref{lem:case1}.
\end{proof}

\section{Binary PFA}

The following theorem shows that the $\lambda$-reachability and emptiness problems are NP-hard for binary PFA of \emph{polynomial ambiguity} with \emph{commuting} transition matrices (and the matrices can be assumed fixed in the case of $\lambda$-reachability and nonstrict emptiness). The emptiness problem for non-commutative binary PFA over $25$ states is known to be undecidable, at least over exponentially ambiguous PFA \cite{Hir}. Emptiness is also undecidable for exponentially ambiguous commutative PFA, although with many more states and a larger alphabet \cite{Bell19}. 

\begin{thm}\label{thm:NPbin}
The $\lambda$-reachability and emptiness problems are NP-hard for binary  probabilistic finite automata of polynomial ambiguity with commuting matrices of dimension $9$ for $\lambda$-reachability, $37$ for nonstrict emptiness, and $40$ for strict emptiness. Moreover, the matrices can be assumed \emph{fixed} for the $\lambda$-reachability and nonstrict emptiness problems.
\end{thm}
\begin{proof}
We use a reduction from the solvability of binary quadratic Diophantine equations. Namely, given an equation of the form $ax^2 + by - c = 0$, where $a, b, c \in \mathbb{N}$, it is NP-hard to determine if there exists $x, y \in \mathbb{N}$ satisfying the equation \cite{MA78}. We begin with the $\lambda$-reachability problem before considering the emptiness problem.

\noindent {\bf $\lambda$-Reachability reduction.}
Let $A = \begin{pmatrix} 1 & 1 \\ 0 & 1 \end{pmatrix}$ and note that $A^k = \begin{pmatrix} 1 & k \\ 0 & 1 \end{pmatrix}$ and that $(A \otimes A)^k_{1,4} = (A^k \otimes A^k)_{1, 4} = k^2$. We form a weighted automaton\footnote{For our purposes here, by a \emph{weighted automaton} we simply mean an automaton whose initial vector, final vector, and transition matrices are over nonnegative integers.} $W_1$ on binary alphabet $\Sigma = \{h, g\}$ in the following way to encode $ax^2 + by$ (we will deal with $c$ later). Let $W_1 = (u_1, \phi, v_1)$ where $u_1, v_1 \in \mathbb{N}^7$ and $\phi: \Sigma^* \to \mathbb{N}^{7 \times 7}$. We define $u_1 = (a, 0, 0, 0, b, 0, 0)^T$, $v_1 = (0, 0, 0, 1, 0, 1, 0)^T$ and $\phi(\ell) = \frac{1}{4}\phi'(\ell)$ for $\ell \in\{h, g\}$ with
\[
\phi'(h) = \left( \begin{array}{@{}c|l@{}} (A \otimes A) \oplus I_2 &  t_1 \\ \hline {\bf 0}^6 & 4 \end{array}\right), \quad
\phi'(g) = \left(\begin{array}{@{}c|l@{}} I_4 \oplus A & t_2 \\ \hline {\bf 0}^6 & 4 \end{array}\right),
\]
with ${\bf 0}^k = (0, 0, \ldots, 0) \in \mathbb{N}^k$, $t_1 = (0, 2, 2, 3, 3, 3)^T$ and $t_2 = (3, 3, 3, 3, 2, 3)^T$. We see then that each row of $\phi'(\ell)$ is nonnegative and sums to $4$, thus $\phi(\ell)$ is stochastic for $\ell \in \{g, h\}$. Furthermore, by the mixed product property of the Kronecker product, we see that $((A \otimes A) \oplus I_2)^x = (A^x \otimes A^x) \oplus I_2$ and $( I_4 \oplus A)^y =  I_4 \oplus A^y$ for $x, y \in \N$ and thus by the block upper triangular structure of $\phi'(h), \phi'(g)$, we see that
\[
\phi'(h^xg^y) = \left( \begin{array}{@{}c|l@{}} (A^x \otimes A^x) \oplus A^y &  t_{xy} \\ \hline {\bf 0}^6 &  4^{x+y} \end{array}\right),
\]
where $t_{xy}$ is a nonnegative vector maintaining the row sum at $4^{x+y}$. We now see that \begin{eqnarray}
u_1^T \phi(h^xg^y) v_1 = \frac{ax^2 + by}{4^{x+y}} \label{eqnn1}
\end{eqnarray}

We define a second weighted automaton $W_2 = (u_2, \psi, v_2)$ with $u_2=(c, 0)^T$, $v_2=(0, 1)^T$ and $\psi:\Sigma^* \to \N^{2 \times 2}$ with $\psi(\ell) = \frac{1}{4}\psi'(\ell)$ for $\ell \in\{h, g\}$ defined thus: $
\psi'(h) = \psi'(g) = \begin{pmatrix} 1 & 3 \\ 0 & 4 \end{pmatrix} $. 
We therefore see that
\begin{eqnarray}
u_2^T \psi(h^xg^y) v_2 = \frac{c(4^{x+y}-1)}{4^{x+y}} = c(1-\frac{1}{4^{x+y}})  \label{eqnn2}
\end{eqnarray}

We now join $W_1$ and $W_2$ into a $9$-state PFA $\mathcal{P} = (u, \gamma, v)$ where $u = \frac{1}{a+b+c} [u_1 | u_2]$, $v = [v_1 | v_2]$ and $\gamma(\ell) = \phi(\ell) \oplus \psi(\ell)$. Combining Eqns~(\ref{eqnn1}) and (\ref{eqnn2}) we see that
\begin{eqnarray}
u^T \gamma(h^xg^y) v & = & \frac{1}{a+b+c}\left(\frac{ax^2 + by}{4^{x+y}} + c(1-\frac{1}{4^{x+y}})\right) \nonumber \\
& = & \frac{1}{a+b+c}\left(c + \frac{ax^2+by-c}{4^{x+y}}\right) \label{binquad}
\end{eqnarray}
which equals $\frac{c}{a+b+c}$ if and only if $ax^2+by-c = 0$. Note that $\gamma(h)$ and $\gamma(g)$ commute by their structure since clearly $(A \otimes A) \oplus I$ and $I_4 \oplus A$ commute, giving $(A \otimes A) \oplus A$ in both cases (as a consequence of the mixed product properties of Lemma~\ref{kronprop}) and the rightmost vector of the matrix simply retains the row sum at $1$ for such a product since the matrices are stochastic. Both $\gamma(h)$ and $\gamma(g)$ are upper-triangular thus $\mathcal{P}$ is polynomially ambiguous.

\noindent {\bf Nonstrict Emptiness reduction.}
We now show the proof of the emptiness problem. We showed that the $\lambda$-reachability problem is NP-hard by deriving a PFA $\pfa$ over the binary alphabet $\{h, g\}$ such that $\pfa(h^xg^y)$ is given by Eqn.~\ref{binquad}. We note however that a non solution to $ax^2+by-c = 0$ can be positive or negative and thus we may be above or below the threshold $\frac{c}{a+b+c}$. This encoding thus cannot be used to show the NP-hardness of the \emph{emptiness} problem.

Instead, we can use a similar encoding of the quartic polynomial given by
 $(ax^2+by - c)^2 = a^2 x^4 + 2 a b x^2 y + b^2 y^2 + c^2 - 2 a c x^2 - 2 b c y$ with $a, b, c \in \N$. Note that we arranged the four positive terms first, followed by the two negative terms. Clearly $(ax^2+by - c)^2$ is nonnegative and equals zero if and only if $ax^2+by - c = 0$. We will derive a PFA $\pfa_2$ such that
 \[
 \pfa_2(h^x g^y) = \frac{1}{z}\left((2ac + 2bc) +  \frac{1}{16^{x+y}} (ax^2 + by + c )^2 \right),
 \]
 where $z = a^2+2ab+b^2+c^2+2ac+2bc$,
 with the property that $\pfa_2(h^x g^y) \geq \frac{2ac+2bc}{z}$ with equality if and only if $(ax^2+by - c)^2 = 0$ which is NP-hard to determine. To this end, we compute the following four matrices $\{H_+, G_+, H_-, G_-\}$, the idea being that $H_+$ and $G_+$ will be used to compute the positive four terms and $H_-$ and $G_-$ will compute the negative terms:
 \begin{eqnarray*}
 H_+ & = & \underbrace{(A \otimes A \otimes A \otimes A)}_{x^4} \oplus \underbrace{(A \otimes A \otimes I_2)}_{x^2y} \oplus \underbrace{(I_2 \otimes I_2)}_{y^2} \oplus \underbrace{1}_{1} \\ 
 G_+ & = & \underbrace{(I_2 \otimes I_2 \otimes I_2 \otimes I_2)}_{x^4} \oplus \underbrace{(I_2 \otimes I_2 \otimes A)}_{x^2y} \oplus \underbrace{(A \otimes A)}_{y^2} \oplus \underbrace{1}_{1} \\
  H_- & = & \underbrace{(A \otimes A)}_{x^2} \oplus \underbrace{I_2}_{y}\\
 G_- & = & \underbrace{(I_2 \otimes I_2)}_{x^2} \oplus \underbrace{A}_{y} 
 \end{eqnarray*}
 and by the mixed product property of Kronecker products of Lemma~\ref{kronprop}),
 \begin{eqnarray*}
 H_+^x G_+^y & = & (A^x \otimes A^x \otimes A^x \otimes A^x) \oplus (A^x \otimes A^x \otimes A^y) \oplus (A^y \otimes A^y) \oplus 1  \\
  H_-^x G_-^y & = & (A^x \otimes A^x) \oplus A^y
 \end{eqnarray*}
 
  Note that $H_+^x G_+^y$ and $H_-^x G_-^y$ each contain the positive and negative (respectively) term of $(ax^2+by - c)^2$, excluding the coefficients, e.g. $(H_+^x G_+^y)_{1, 16} = x^4$ and $(H_+^x G_+^y)_{17, 24} = x^2y$ etc. Note also that $H_+G_+ = G_+H_+$ and $H_-G_- = G_-H_-$ which also follows from the mixed product properties and thus matrices $\{H_+, G_+\}$ and $\{H_-, G_-\}$ commute.
 
  As before, we may now increase the dimension of each matrix $\{H_+, H_-, G_+, G_-\}$ by $1$ to ensure a common row sum (of $16$ in this case) by adding a new column on the right hand side of each matrix, and then divide each matrix by this common value to give $\{H'_+, H'_-, G'_+, G'_-\}$ so that each of these matrices is row stochastic. Matrices $\{H'_+, G'_+\}$ and $\{H'_-, G'_-\}$ still commute since this change only has an effect on the final column of the matrix.
  
  We now show how to handle each term of $(ax^2+by - c)^2$. We first handle the positive terms. We define $u_1 = (a^2, 0, \ldots, 0)^T \in \Q^{16}$, $u_2 = (2ab, 0, \ldots, 0)^T \in \Q^8$, $u_3 = (b^2, 0, 0, 0)^T \in \Q^{4}$ and $u_4 = c^2$ and then let $u_+ = [u_1 | u_2 | u_3 | u_4 | 0] \in \Q^{30}$. We let $v_1 = (0, \ldots, 0, 1)^T \in \Q^{16}$, $v_2 = (0, \ldots, 0, 1)^T \in \Q^{8}$, $v_3 = (0, 0, 0, 1)^T \in \Q^{4}$ and $v_4 = 1$, and let $v_+ = [v_1 | v_2 | v_3 | v_4 | 0] \in \Q^{30}$. We then see that 
  \begin{eqnarray} & &
  u_+^T (H'_+)^x (G'_+)^y v_+ \nonumber \\ &= & \frac{1}{16^{x+y}} \left( u_1^T (A^x \otimes A^x \otimes A^x \otimes A^x) v_1 + 
  u_2^T (A^x \otimes A^x \otimes A^y) v_2 +
  u_3^T (A^y \otimes A^y) v_3 +
  u_4^T v_4
  \right) \nonumber \\
   & = & \frac{1}{16^{x+y}}\left(a^2x^4 + 2abx^2y + b^2y^2 + c^2 \right) \label{expans1}
  \end{eqnarray}
  
  We next handle the negative terms, which is essentially accomplished by switching final and non-final states in the final state vectors to follow. Define $u_5 = (2ac, 0, 0, 0)^T \in \Q^4$ and $u_6 = (2bc, 0)^T \in \Q^2$ and let $u_- = [u_5 | u_6 | 0] \in \Q^{7}$. We let $v_5 = (0, 0, 0, 1)^T \in \Q^4$ and $v_6 = (0, 1)^T \in \Q^2$. Define $v_- = [v_5 | v_6 | 0] \in \Q^{7}$. We then see that
   \begin{eqnarray}
   & &
  u_-^T (H'_-)^x (G'_-)^y ({\bf 1} - v_-) \nonumber \\
   & = & (2ac + 2bc) - \frac{1}{16^{x+y}} \left( u_5^T (A^x \otimes A^x)  v_5 +
   u_6^T A^y v_6 + 0 \right) \nonumber
   \\
   & = & (2ac + 2bc) - \frac{1}{16^{x+y}} \left( 2acx^2 + 2bcy \right), \label{expans2}
    \end{eqnarray}
    where ${\bf 1} = (1, 1, \ldots, 1)^T \in \Q^{7}$. We used here the fact that $X{\bf 1} = {\bf 1}$ for a row stochastic matrix $X$.
  We finally define that $H = H'_+ \oplus H'_- \in \mathbb{Q}^{37 \times 37}$ and $G = G'_+ \oplus G'_- \in \mathbb{Q}^{37 \times 37}$, both of which are row stochastic and commute,  and let $u_\star = \frac{[u_+ | u_-]}{z} \in \mathbb{Q}^{37}$ and $v_\star = [v_+ | ({\bf 1} - v_-)] \in \mathbb{Q}^{37}$, with $z = a^2+2ab+b^2+c^2+2ac+2bc$ to normalise vector $u_\star$. We see then that $u_\star$ is a stochastic vector as required. We define the PFA $\pfa_2 = (u_\star, \{H, G\}, v_\star)$ and we can now compute that 
  \begin{eqnarray*} &  & \mathcal{P}_2(h^xg^y) = 
  u_\star^T H^x G^y v_\star\nonumber \\
  & = & u_\star^T (H_+'^x G_+'^y \oplus H_-'^x G_-'^y) v_\star \nonumber \\ & = &  \frac{1}{16^{x+y}} \left(\frac{[u_+ | u_-]}{z}^T \left(\begin{array}{c|c}
  \begin{array}{c|c} H_+^xG_+^y & * \nonumber \\ \hline {\bf 0} & 16^{x+y} \end{array} & 
  {\bf 0} \nonumber  \\ 
  \hline
 {\bf 0} & 
  \begin{array}{c|c} H_-^xG_-^y & * \\ \hline {\bf 0} & 16^{x+y} \end{array} \end{array}\right) [v_+ | ({\bf 1} - v_-)] \right) \nonumber \\
  \end{eqnarray*}

  \begin{eqnarray}
  & = & \frac{1}{z16^{x+y}} \left( u_+^TH_+^xG_+^yv_+ + u_-^TH_-^xG_-^y({\bf 1} - v_-)  \right)  \nonumber \\
  & = & \frac{1}{z} \left( u_+^T(H'_+)^x(G'_+)^yv_+ + u_-^T(H'_-)^x(G'_-)^y({\bf 1} - v_-)  \right) \nonumber \\
  & = & \frac{1}{z}\left( (2ac + 2bc) + \frac{1}{16^{x+y}}\left(a^2x^4 + 2abx^2y + b^2y^2 + c^2 \right) - \frac{1}{16^{x+y}} \left( 2acx^2 + 2bcy \right) \right) \nonumber \\ 
  & = & \frac{1}{z}\left((2ac + 2bc) +  \frac{1}{16^{x+y}} (ax^2 + by - c )^2 \right) \label{finalEq}
  \end{eqnarray}
where $*$ denote the column vectors used to ensure row sums of $16^{x+y}$ and ${\bf 0}$ denotes zero matrices of appropriate sizes. We also used Eqns~(\ref{expans1}) and (\ref{expans2}). 

Since $(ax^2 + by - c )^2$ is nonnegative, we see that $u_\star^T H^x G^y v_\star \geq \frac{2ac+2bc}{z}$ with equality if and only if $(ax^2 + by - c )^2 = 0$, which is NP-hard to determine. Therefore using cutpoint $\lambda = \frac{2ac+2bc}{z} \in \Q \cap [0, 1]$ means the (nonstrict) emptiness problem is NP-hard (i.e.~does there exist $x, y \in \N$ such that $u_\star^T H^x G^y v_\star \leq \lambda$ is NP-hard). As before, matrices $H$ and $G$ are upper-triangular and commute by their structure, and therefore the result holds.

\noindent {\bf Strict Emptiness reduction.}
Finally we show how to handle the strict emptiness problem. We proceed with a technique inspired by \cite{GO10}. By (\ref{finalEq}), if $P_2(h^xg^y) = u_\star^T H^x G^y v_\star \neq \frac{1}{z}(2ac + 2bc)$, then $
u_\star^T H^x G^y v_\star 
 \geq \frac{1}{z}\left((2ac + 2bc) +  \frac{1}{16^{x+y}} \right)
 $ 
 therefore $P_2(h^xg^y) \leq \frac{1}{z}(2ac + 2bc)$ if and only if $P_2(h^xg^y) < \frac{1}{z}\left((2ac + 2bc) +  \frac{1}{16^{x+y}} \right)$.
 
   Let us adapt $\pfa_2$ in the following way to create a new PFA $\pfa_3$. Note that $\pfa_2$ has $6$ initial states (by $u_\star$). We add three new states to $\pfa_3$, denoted $q_0, q_F$ and $q_*$. State $q_0$ is a new initial state of $\pfa_3$ which, for any input letter, has probability $\frac{1}{2\cdot 6}$ of moving to each of the $6$ initial states of $\pfa_2$ and probability $\frac{1}{2}$ to move to new state $q_F$. State $q_F$ is a new final state that remains in $q_F$ for any input letter with probability $1-\frac{1}{16z}$ and moves to a new non-accepting absorbing sink state $q_*$ with probability $\frac{1}{16z}$. We now see that for any $a \in \{h, g\}$:
\[ \pfa_3(aw) = \frac{1}{2}\pfa_2(w) + \frac{1}{2}\left(1-\frac{1}{16^{|w|}z^{|w|}}\right)
\]
If there exists $w_1 = h^xg^y$ with $x, y \geq 0$ such that  $\pfa_2(w_1) \leq \frac{1}{z}(2ac + 2bc)$ then $\pfa_2(w_1) = \frac{1}{z}(2ac + 2bc)$ and thus:
\[\pfa_3(aw_1) = \frac{1}{2}\left(\frac{1}{z}(2ac + 2bc)\right) + \frac{1}{2}\left(1-\frac{1}{16^{|w_1|}z^{|w_1|}}\right) < \frac{1}{2}\left(\frac{1}{z}(2ac + 2bc)+ 1\right).
\]
For any $w_2 = h^xg^y$ with $x, y \geq 0$ such that $\pfa_2(w_2) > \frac{1}{z}(2ac + 2bc)$ then $\pfa_2(w_2) \geq \frac{1}{z}(2ac + 2bc) + \frac{1}{16^{x+y}}$ by~(\ref{finalEq}). Thus: 
\[
\pfa_3(aw_2) \geq \frac{1}{2}\left(\frac{1}{z}(2ac + 2bc) + \frac{1}{16^{|w_2|}}\right)+ \frac{1}{2}\left(1-\frac{1}{16^{|w_2|}z^{|w_2|}}\right) > \frac{1}{2}\left(\frac{1}{z}(2ac + 2bc) + 1\right).
\]
Thus determining if there exists $w = h^xg^y$ such that $\pfa_3(w) < \frac{1}{2}\left(\frac{1}{z}(2ac + 2bc) + 1\right)$, i.e. the strict emptiness problem for $\pfa_3$ on cutpoint $\frac{1}{2}\left(\frac{1}{z}(2ac + 2bc) + 1\right)$, is NP-hard. The modifications to $\pfa_2$ retain polynomial ambiguity since $q_0$ and $q_F$ have no incoming (non self looping) edges and $q_*$ has no outgoing edges, therefore property EDA does not hold. Commutativity of the PFA is unaffected since $\mathcal{P}_{3}$ is identical to $\mathcal{P}_2$ except for adding three new states, behaving identically for both input letters. Note that $\pfa_3$ has $37+3 = 40$ states.
\end{proof}

\section{Conclusion}
We considered the emptiness and $\lambda$-reachability problems for unary and binary probabilistic finite automata of bounded ambiguity. Our main result is an EXPTIME algorithm to solve $\lambda$-reachability and emptiness for unary polynomially ambiguous PFA. We may note that our procedure generates a polynomial size certificate for these problems, however the verification step takes EXPTIME. We also note that the $\lambda$-reachability and emptiness problems are NP-hard for unary finitely ambiguous PFA. It would be interesting to close the gap between the upper and lower bounds for this problem. We also showed NP-hardness results for  $\lambda$-reachability ($9$ states) and strict ($40$ states) and nonstrict ($37$ states) emptiness for binary polynomially ambiguous PFA with commutative transition matrices but a PSPACE-hardness result in this setting appears challenging. A further avenue of research is to reduce the number of states for which NP-hardness holds.
\bibliographystyle{alphaurl}
\bibliography{refs}

% ======================================

\end{document}